\documentclass[a4paper]{article}
\usepackage{amsmath,amsfonts,amsthm,amssymb,amsbsy,
amstext,amscd,amsxtra,multicol}
\usepackage{graphicx}
\def\mpfile#1#2{\includegraphics{#1-#2-mps.pdf}} 

\title{Distributed protocols for spanning tree construction and
leader election} 
\author{I.~M.~Khuziev\thanks{Moscow Institute of Physics and
    Technology (State University). Partially supported
by the RFBR grant 14--01--00641}\and M.~N.~Vyalyi\thanks{Dorodnicyn Computing Centre, FRC CSC RAS;
Moscow Institute of Physics and Technology (State University);
National Research University Higher School of Economics.
The~study has been funded by the Russian Academic Excellence Project `5-100'}} 
\date{}

\newtheorem{theorem}{Theorem}
\newtheorem{lemma}{Lemma}

\newtheorem{prop}{Proposition}

\newtheorem{cor}{Corollary}
\theoremstyle{definition}
\newtheorem{definition}{Definition}
\newtheorem{remark}{Remark}
\let\ld\lambda

\def\mess{\mathop{\mathrm{message}}\nolimits}
\def\lcp{\mathop{\mathrm{lcp}}\nolimits}
\def\m{{\boldsymbol m}}

\def\bin{\mathop{\mathrm{bin}}\nolimits}

\def\lseq{\preceq_{\text{\upshape\tiny PL}}}
\def\ls{\prec_{\text{\upshape\tiny PL}}}
\def\gs{\succ_{\text{\upshape\tiny PL}}}
\def\lex{\prec_{\text{\upshape\tiny lex}}}

\def\root{\mathop{\mathtt{root}}\nolimits}

\begin{document}
\maketitle
\begin{abstract}
  We present fast deterministic distributed protocols in synchronous
  networks for leader election and spanning tree construction.  The
  protocols are designed under the assumption that nodes in a network
  have identifiers but the size of an identifier is unlimited. So time
  bounds of protocols depend on the sizes of identifiers.
  We present fast protocols running in time $O(D\log
  L+L)$, where $L$ is the size of the minimal
  identifier and $D$ is the diameter of a network.
\end{abstract}

We study deterministic distributed protocols for leader election
and 
 spanning tree construction in a synchronous network. These
problems play important role in distributed computation theory (e.g., see
books~\cite{L96, W14}).

It is well-known that for  anonymous networks (processors are
indistinguishable) there are no deterministic protocols solving these
problems. So we assume that processors in a network (\emph{nodes})
have unique names (\emph{identifiers}). Under this assumption, a
natural way for 
leader election is \emph{the minimal identifier
  broadcast}: if all nodes in a network know the value of the minimal
identifier then the node having this identifier is elected.

The previous study of
these problems   was based on the assumption that the
identifiers are rather short (say, the nodes are numbered in the range
$[1,\dots,V]$, where $V$ is the size of a network)  and messages are
rather long (the message size is $O(\log V)$). In these settings 
near-optimal protocols are known for these problems and  for a more
general problem of minimum spanning tree construction~\cite{PR00, W14}.

We address to an intermediate situation between the anonymous case and the
well-numbered case. Namely, we assume that the nodes have 
identifiers but there is no a priori bound on the size of the
identifiers. Also we restrict the communication speed by $O(1)$ bit
per 
message. We call this model \emph{the unbounded identifier model}.
Here we consider deterministic protocols only and do not analyze possible
errors in the process of information transmission.

In this model bounds for the bit complexity of protocols (the total number of
bits sent) are known.  For asynchronous ring and
chain  networks rather close upper and lower  bounds of the bit
complexity  were obtained in~\cite{DSR08}. In~\cite{KhV15} 
protocols with near-optimal bit complexity bounds were presented  for
arbitrary synchronous networks. More exactly, for each
monotone unbounded function   $g(\cdot)$ there is a 
 protocol that constructs a spanning tree and sends $O(E g(V))$ bits,
where $E$ is the number of links in a network, $V$ is the number of
nodes. These protocols construct rooted trees and the root is the node
with the minimal identifier. So the protocols also elect a leader
in a network.

In this paper we are interested in time bounds of the protocols in the
unbounded identifier model.
In the unbounded identifier model the running time of a protocol
depends  on the  sizes of 
identifiers. We denote by $L$  the length of the shortest identifier.
(We always assume that a shorter identifier is lesser.)

Straightforward modifications of known protocols  for the
unbounded identifier model run in time  $O(DL)$.
Here
$D$ is the diameter of a~network. We
present better protocols. 

We present the message terminating protocol and the processor terminating
protocol for leader election and spanning tree construction.  The
definitions for these types of protocols are taken from the
paper~\cite{IR}.  The message termination means that at some moment of
time all nodes are in a \emph{sleep state}. In such a~state a~node does
not send  messages but   can change the state after receiving a
suitable message. The processor termination means that at some moment
of time all nodes are in the \emph{final state}. This state cannot be
changed by any message. In other words, a node in the final state is
off-line. Processor terminating is a stronger property. It is
important, e.g., for compositions of protocols.

The following two results are our main contribution.

\begin{theorem}\label{th:m-t-prot}
  There exists a message terminating protocol that broadcasts the minimal
identifier and  constructs a rooted spanning tree
in time $O(D\log L+L)$, where $L$ is the length of the
minimal identifier and $D$ is the diameter of the network.
\end{theorem}
\begin{theorem}\label{th:p-t-prot}
  There exists a processor terminating protocol that broadcasts the
  minimal identifier and constructs a rooted spanning tree in time
  $O(D\log L+L)$, where $L$ is the length of the minimal identifier and
  $D$ is the diameter of the network.
\end{theorem}

These protocols are modifications of flooding and echo protocols
designed for the well-numbered model. To speed up a flooding protocol
in the unbounded identifier model, we introduce a suitable encoding of
identifiers (see Section~\ref{id->key} below) and add a correction
phase to the process of information transmission.

Note also that the protocols presented in this paper remain correct
in the asynchronous case. But time bounds are valid for the
synchronous case only.

The rest of paper is organized as follows. In Section~\ref{defs} we
give formal definitions of our model and the problems solved by
protocols. Section~\ref{informal} contains an overview of our
construction. 
Section~\ref{min-broadcasting} contains the proof of
Theorem~\ref{th:m-t-prot}: a~detailed description of the
message terminating protocol that broadcasts the minimal identifier in a
network, the proof of  correctness, the upper  bound of the running
time and  spanning tree construction based on  this protocol.
In Section~\ref{finalizing} we present the proof of
Theorem~\ref{th:p-t-prot} based on a suitable modification of an echo
protocol combined with the previous protocol that broadcasts the
minimal identifier.

\section{Definitions}\label{defs}

The unbounded identifier model were introduced in~\cite{KhV15}. Here
we adopt this model with minor changes in formal definitions. 

We consider point-to-point synchronous networks and distributed
protocols in them. Nodes of a network communicate by sending messages
through bi-directed links. All nodes perform the same communication
protocol. 
Communication
speed is bounded  by $O(1)$. The nodes are distinguishable but
the amount of information to distinguish them may be arbitrary large.

In this section  we provide formal definitions for this model. 

A \emph{network} $G(V,E)$ is a connected graph, where $V$ is the set
of nodes and $E$ is the set of bi-directed links.  The graph is
equipped by a set of functions. An injective function $I\colon V\to
\{0,1\}^*$ assigns the unique identifier $I(v) $ to a node~$v$.  Links
of each node are ordered by index functions $n_v$. Let $d(v)$ be a
degree of a node in $G(V,E)$. Then the range of the index function
$n_v\colon\{1,\dots,d(v)\}\to E$ is exactly the set of edges incident
to $v$ in $G(V,E)$.

We assume that initially a node has only local information about a network. It
knows the number of its links (the degree $d(v)$),
can distinguish
links by   its index function, and knows its identifier. 

This model uses discrete time. Time moments are numbered by
nonnegative integers. At a moment~$t$ a node $v$ receives messages
from all its neighbors. The messages are bit strings, their lengths
are bounded by a constant specified by a protocol. The empty string
$\ld$ is allowed. It indicates absence of a~message.

The list of messages received by the node  $v$ at a moment~$t$ is
denoted as
\[
\m_t(v) = (m_1(t,v), \dots, m_{d(v)}(t,v)), 
\]
the messages are ordered by the index function.

Usually states of nodes are used in protocol descriptions. But it is
unnecessary and we avoid states in the definitions. 
One may take a~history of communication in a node as a state of a node.
\emph{A~history of
communication} $H(t,v)$ in a node $v$ at a moment $t$ is the list
$(I(v),\m_1(v), \dots,\m_t(v)) $. We assume that messages depend on
the history of communication only. (So the protocols are deterministic.)

By definition, a \emph{protocol} is the message function $\mess(t,H)$ 
that determines the messages sent during the current round of
communication. The first   argument in the message function is an index of a link in
the link enumeration. The second argument is the history of communication.

The message function may be arbitrary. It means that
nodes are computationally unlimited.  Nevertheless, it is worth to
mention that in the protocols presented in this paper the message
functions are easy to compute.

Initially, at time $t=0$, 
the history of
communication in a node $v$ is $H(v,0)=(I(v))$. 
Then the configuration of the network is changed according the following
rules:
\begin{equation}\label{evolution}
  \begin{aligned}
    m_{i}(t,v)&= \mess( j, H(t-1,s)),\\
    H(t,v)& = (H(t-1,v), (m_{i}(t,v): 1\leq i\leq d(v))),
  \end{aligned}
\end{equation}
where the link $(s,v)\in E$ has the index $i$ in the node $v$ and the
index $j$ in the node $s$. (Recall that no communication errors are
allowed in this model.)

Note that the rules are the same for all nodes in the network. But
 behavior of a node may depend on a
moment of time due to the synchronization assumption.

We say that a node $v$ is in the
\emph{final state} at moment~$t$, if
\[
  \mess(i,(H(t,v), H'))=\ld
\]
for all  $i$, $H'$. 
It means that in the final state a node is not
communicating with 
a network: it sends nothing and ignores all incoming messages.

We say that a node $v$ is in 
 a \emph{sleep} state at moment~$t$, if 
\[
  \mess(i,(H(t,v),\Lambda,\Lambda,\dots,\Lambda )=\ld,
\]
where $\Lambda$ is the list composed of the empty strings. 
Thus a node in a sleep
state does not generate messages and does not change its state if
it has not received messages from its neighbors at the previous moment
of time. It implies that communication finishes if all nodes are in
sleep states. But a node in a sleep state can be `awaken' by a
suitable message. 

\begin{definition}\upshape
A protocol is \emph{message terminating} if for any network $G(V,E)$,
 $I(v)$,  $n_v(i)$  
all nodes are in a sleep state 
at some moment~$t$.
\end{definition}

\begin{definition}\upshape
A protocol is \emph{processor terminating} if for any $G(V,E)$,
$I(v)$,  $n_v(i)$  
 all nodes are in the final state  at some moment~$t$.
\end{definition}

The running time $T$ of the protocol in the network $G(V,E)$, $I(v)$,
 $n_v(i)$
 is the first
moment of time when all nodes are in the final (a sleep) state. 

Now we are going to define a useful result of a
protocol. Loosely speaking,  it is an arbitrary function of a history of
communication. 

We say that a protocol $\mess(i,H)$  \emph{broadcasts the minimal
identifier} if it terminates in any network
and there is a function $M(H)$ such that
$M(H(v,T))=I_{\min}$ for all $v\in V$. Here $T$ is the running time of
the protocol and $I_{\min}$ is the minimal identifier in the network
with respect to \emph{the shortlex order} $\lex$: compare strings by
the length and in the case of equal lengths compare them
lexicographically.

We say that a protocol $\mess(i,H)$ \emph{elects a leader} if it
terminates in any network and there is a function $\ell(H)$ such that
(i) $\ell(H)\in\{0,1\}$ and (ii) after the protocol finishes, there is
exactly one node (the leader) with the value~$\ell=1$:
$\ell(H(w,T))=1$ for the leader $w$ and $\ell(H(v,T))=0$ for all $v\ne w$.

It follows from the definitions that if a protocol broadcasts the minimal
identifier then it also elects a leader: the value  $\ell(v)$ is
determined by  $I_{\min}$ and $I(v)$ in the straightforward way.

We say that a protocol $\mess(i,H)$ \emph{constructs a rooted
spanning tree} if it terminates  in any network and there is a function
$\pi(H)$ such that (i) $\pi(H(v,T))\in V\cup\{\root\}$; (ii) there is
exactly one node $r$ such that $\pi(H(r,T))=\root$; (iii) for all
$v\in V$ if $\pi(H(v,T))\ne\root$ then $(v, \pi(H(v,T)))\in E$ and
(iv) the links $(v, \pi(H(v,T)))$, where $\pi(H(v,T))\ne\root$, form a
spanning tree in the network.
 
\begin{remark}
  The definition of protocols constructing a spanning tree looks rather
  artificially. It is used to simplify  formal arguments. Note in advance that
  an orientation from a child to a parent appears in our protocols
  naturally as well as the list of childs. So we skip the list in
  the definition. 

  Speaking of rooted trees, we always assume edge orientation towards
  the root.
\end{remark}

\section{An overview of the protocols}\label{informal}

In this section we discuss informally the protocols presented below.

The basic idea is very simple: to broadcast information as fast as
possible. Once a~node obtains some useful information, it immediately
resends it to all its neighbors.

We illustrate the idea in a simpler situation.  Suppose that initially
there is the unique \emph{active} node in a network (a leader).  It
starts communication and sends its identifier to neighbors bit by
bit. An initially nonactive node awakes after receiving a non-empty
message and chooses a sender of this message as a parent.  Then the
node also transmits to its neighbors bits received from the parent.
In this way we obtain a protocol that constructs a rooted tree in time
$O(D+L)$, where $L$ is the size of the active node identifier.  The
root of the tree is the initially active node.  Such protocols are
called \emph{flooding protocols}.

In absence of a leader (the initially active node) the
straightforward realization of this idea takes more time. We mean the
following rules: nodes send identifiers to the neighbors, compare
received identifiers and transmit the current minimal value. In this way we obtain a sort of flooding protocol to
broadcast the minimal identifier. The protocol is easily transformed
to a protocol constructing a spanning tree. 
The running time of this protocol in
our settings is $O(DL)$. Essentially, it
is the well-known Perlman protocol~\cite{P85} adjusted for
our model.

To speed up this flooding protocol we will use two ideas. 

At first, we encode identifiers in such a way that for encoded strings
the result of comparison of string's prefixes does not change after
receiving the suffixes of the strings. It gives a way to compare long
strings by their short prefixes. Thus, our protocols transmit
\emph{keys}, i.e. encoded identifiers.

The second idea is to transmit the current information about the
minimal key although it might be incorrect. When incorrectness is
revealed, a correction message is sent instead of re-sending the whole
information about the minimal key. Of course,    corrections
make the running time longer. In Section~\ref{min-broadcasting} we
give the exact rules for the protocol and prove the  time
bound given in Theorem~\ref{th:m-t-prot}. 

In this way we get a message terminating protocol. After receiving a
set of keys, a~node can not exclude a possibility that
somewhere in a network there is a smaller key.

To make a protocol processor terminating requires additional
work.  Thus nodes should communicate and establish
the fact of termination of the key distribution process.  To make this
communication faster, the nodes may use a spanning tree built by the
message terminating protocol. For this purpose we implement a~protocol
based on the idea of  echo protocol (see, e.g.~\cite{W14}) adjusted
to our model.  Broadcasting messages through a tree requires a time
proportional to the diameter of the tree. 
Thus we come to the time bound given in Theorem~\ref{th:p-t-prot}: the
tree depth does not exceed the running time of the message terminating
protocol. Details are presented in Section~\ref{finalizing}.

\section{Message terminating protocol for the minimal identifier
  broadcasting}\label{min-broadcasting} 

\subsection{From identifiers to keys}\label{id->key}

We encode the identifiers by the map
\begin{equation*}
 K\colon x\mapsto 1^{k+1}0^{2^k-|x|}1x,\quad\text{where } k=\lceil\log
 |x|\rceil+1,
\end{equation*}
on the set $\{0,1\}^*$ of binary strings. 
Hereinafter we use the standard notation: $|x|$ is the length of a
  binary string $x$,  $a^k$ is a string $a$ repeated $k$
  times. All logarithms are binary. For the empty string $\ld$ we set
  $K(\ld) = 101$. It corresponds to the general rule with $k= 0$.

The image $c(\{0,1\}^*)$ is the set of \emph{keys}.
 
It is easy to see  several  simple and useful properties of this encoding:
\begin{itemize}
\item [--] if $|x|=|y|$ then $|K(x)|=|K(y)|$;
\item [--] the length of $x$ is determined by a prefix
  $1^{k+1}0^{2^k-|x|}1$ of the key $K(x)$;
\item [--] the encoding is injective (an identifier can be restored from its key);
\item [--] any key is not a proper prefix of another key;
\item [--] $|K(x)|=O(|x|)$. 
\end{itemize}

We are interested in asymptotic bounds. So we will ignore
the difference between 
lengths of keys and 
lengths of identifiers.

Keys are compared in \emph{the partial lexicographical order} (PL order). Let
$\lcp(u,v)$ be the longest common prefix of the strings $u$, $v$. Then
by definition $u$ is less than $ v$ in the PL order if
$u=\lcp(u,v)0u'$, $v=\lcp(u,v)1v'$.  We denote this fact by $u\ls v$.

Note that a prefix of a string is not PL comparable with the string. It
justifies the term `partial'. In the other direction, if $\lcp(u,v)$
is a proper prefix of both strings $u$ 
and $v$ then the strings are comparable in the PL order.

We indicate properties of the PL order that will be useful in the
analysis of the protocols.

\begin{prop}\label{greedy-compare}
  If  $u\ls v$  then $uw'\ls vw''$ for all $w'$, $w''$.
\end{prop}

The proposition follows  easily  from the definition of the PL order.

\begin{prop}\label{length-compare}
  If  $u$, $v$ are keys and
  $|u|<|v|$ then   $u\ls v$. Moreover, there are 
  prefixes $u'$ and $v'$ of the keys
  $u$, $v$ such that 
  $u'\ls v'$ and $|u'|=|v'|\leq2+ \log |u|$.
\end{prop}
\begin{proof}
  Note that $\lcp(u,v)=1^{k+1}$, where $|u|=2+k +2^k $. So the prefixes
  $u'=1^{k+1}0$, $v'=1^{k+2}$ of
  the keys $u$, $v$ are PL compared and $u'\ls v'$.

  The second claim
  follows from the inequality $k+2\leq
  2+\log(2+k+2^k)$. 
\end{proof}

\begin{remark}
  There is no linear order on the whole set of binary strings satisfying
  Proposition~\ref{greedy-compare} and
  Proposition~\ref{length-compare}. Thus we need to introduce keys.

  Also note that the PL order is a linear order  on
  the set of the keys. 
\end{remark}

\begin{prop}\label{lex-consistent}
The PL order on the keys is consistent with the shortlex
 order $\lex$ on the identifiers: $x\lex y$ iff $K(x)\ls K(y)$.
\end{prop}
\begin{proof}
  If $|x|=|y|$  then $K(x) = p x$, $K(y)=py$. In this case  $K(x)\ls
  K(y)$ is equivalent to $x\lex y$.

  If $|x|<|y|$ then either $|K(x)|<|K(y)|$  or $|K(x)|=|K(y)|$.  The former
  implies $K(x)\ls K(y)$  by Proposition~\ref{length-compare}.  If
  $|K(x)|=|K(y)|$, then 
  $K(y)=1^{k+1}0^{2^k-|y|}1y$ and $K(x) =1^{k+1}0^{2^k-|y|}0sx $. Thus
  $K(x)\ls K(y)$.

  In the other direction the arguments are similar.
\end{proof}

\begin{prop}\label{min-PL}
  Let $M$ be the set of PL-minimal elements of a set $S$ of strings. Then for any
  pair of strings in $M$ one is a prefix of another.
\end{prop}
\begin{proof}
  The minimal elements form an antichain in the PL order. But the only
  PL uncomparable pairs are (string $u$, a prefix of $u$).
\end{proof}

\subsection{Description of the protocol}

Messages in the protocol consist of the \emph{information bit} and
\emph{control bits}. There will be finitely many control bits.  Thus
the communication speed is $O(1)$.  The value~1 of a~control bit
indicates that a specific event occurs in the process of
communication.

Absence of the information bit in a message is allowed (the case is
indicated by the empty string instead of 1-bit string).
We assume that if the information bit is empty and all control bits
are~$0$, then no message is sent (formally, the empty string is sent).

Each node sends the same messages to all its neighbors in this
protocol. The message is determined by the following data. 

Executing the protocol, a node maintains \emph{the candidate} (to the
minimal key) and the list
of \emph{participants}. Candidates and participants are binary strings.
Participants are  prefixes of the candidates that has been received
by a node from its neighbors. In particular, the number of
participants equals the number of neighbors in each node.

Initially the
candidate of a node is its key and all participants are empty. 

At any moment of time the candidate is the longest PL-minimal string
among the participants and the node's
key. Proposition~\ref{min-PL} guarantees the uniqueness of the
candidate.

At \emph{regular periods of time} a node sends bits of its candidate
one by one until all bits of the candidate are sent.  At
\emph{exceptional periods of time} the participant $p$ sent to the
neighbors of a node is not a prefix of the current candidate $q$ of
the node. During an exceptional period the node sends \emph{a
correction message}: the length of $\lcp(p,q)$ written in binary. (In
other words, a node tells to neighbors how many bits in its data are
valid.)  The end of a correction message is indicated in control bits
as well as the end of a correction message.

It is possible that the candidate is changed during an exceptional period
of time. It might cause that the current correction message becames
incorrect (more bits in the participant are wrong). In this case the
current correction message is aborted (this event is indicated in
control bits) and the node starts transmission
of a new correction message.

If a node has finished  transmission of the candidate it goes to
a~sleep state.

Below we provide the exact rules of the protocol. We need the
following notation.

Let $w[i:j]$ be the substring of a string $w$
that starts at the position $i$ and ends in the position $j$. A~single
$j$th bit of the string is denoted by~$w[j]$.

$K(u)$ is the key of the node $u$, i.e., is a shortcut of
$K(I(u))$. 
$C(t,u)$ is the candidate of the node $u$ after  $t$ steps of the protocol. 
$P(t,u)$ is the participant sent by the node $u$  to all
its neighbors  after  $t$ steps of the protocol. 

To define the rules for exceptional periods we need a notation 
$R(t,u)$ that means the prefix of the correction message that has been
sent to neighbors after  $t$ steps of the protocol. During regular
periods,  $R(t,u)$ is the empty string.

In this protocol a node sends the same message to all its
neighbors. So we indicate an information sent by the node. Just the
same information is received by all neighbors of the node (recall that
there are no   communication errors in our model). 

Thus, an operation of a node at step  $t+1$ of protocol is defined 
by the data
\[
C(t,u), \ P(t,u),\ R(t,u),\quad (P(t,v): (v,u)\in E(G)),
\quad (R(t,v): (v,u)\in E(G)),
\]
i.e., by the candidate, by the prefixes of the candidate and a correction
message sent up to the time, and by participants and prefixes of
correction messages received from neighbors. The node also takes into
account the values of control bits. For simplicity we do not introduce
notation for control bits. Instead, we describe a way to use them.

To form a message at step  $t+1$, the node also uses a correction
message string $L(t,u)$. It is determined by the basic data as
follows: if  $P(t,u)$ is a prefix of  $C(t,u)$ (a regular period),
then   $L(t,u)$ is the empty string. Otherwise (an exceptional period)
$$
L(t,u) = \bin\Big(\big|\lcp(C(t,u),P(t,u))\big|\Big),
$$ 
where $\bin(n)$ is a binary representation of an integer~$n$.

If  $L(t,u)$ is empty, then, if $P(t,u)\ne C(t,u)$, the information bit
  $b(t+1,u)$ of the next message is the next bit of the candidate, i.e.,
\[
b(t+1,u) = C(t,u)\big[|P(t,u)|+1\big];
\]
otherwise, $b(t+1,u)=\ld$. 

If  $L(t,u)$ is non-empty, then the node transmits the following
 message (and indicates in control bits that it is a message of an
 exceptional period): 
\begin{enumerate}
\item[--] if $R(t,u)$ is non-empty and  $R(t,u)$ is a prefix of  $L(t,u)$, then
  the information bit  $b(t+1,u)$ of the message is the next bit of
  the correction message, i.e.,
  \[
  b(t+1,u) = L(t,u)\big[|R(t,u)|+1\big];
  \]
\item [--] if either $R(t,u)$ is empty or $R(t,u)$ is not a prefix of
  $L(t,u)$, then the information bit 
  $b(t+1,u)$ of the message is  $L(t,u)[1]$, in control bits the fact
  of starting correction phase is indicated;
\item[--] a~transmission of the last bit of a correction message is
  indicated in control bits also.
\end{enumerate}

Figures~\ref{regular-message} and~\ref{correction-message} illustrate
the rules determining the value of the information bit of a message.

\begin{figure}
  \begin{minipage}[b]{0.45\textwidth}
    \centerline{\mpfile{message}{1}}
    \caption{Regular period}\label{regular-message}
  \end{minipage}\hfill
  \begin{minipage}[b]{0.45\textwidth}
    \centerline{\mpfile{message}{2}}
    \caption{Exceptional period}\label{correction-message}
  \end{minipage}
  \parfillskip=0pt\par
\end{figure}

Depending on the messages received at moment~$t+1$, the nodes change
the basic data according the following rules.

If  $L(t,u)=\ld$ (a regular period), then 
$$P(t+1,u) = P(t,u)b(t+1,u),\quad R(t+1,u) = \ld.$$ 

If $L(t,u)\ne\ld$ (an exceptional period), then changes in the data
are different:
\begin{enumerate}
\item [--] if  $R(t,u)b(t+1,u)= L(t,u)$ (the end of the exceptional
  period), then
  $R(t+1,u) =\ld$ and $P(t+1,u) =
\lcp(C(t,u),P(t,u)) $;
\item [--] if $R(t,u)b(t+1,u)$ is a proper prefix of
 $L(t,u)$, then   $R(t+1,u)= R(t,u)b(t+1,u)$ and $P(t+1,u) = P(t,u)$;
\item [--]  if  $R(t,u)b(t+1,u)$ is not a  prefix 
 $L(t,u)$, then $R(t+1,u) = b(t+1,u)$ and $P(t+1,u) = P(t,u)$ (it
 follows from the rules determining the messages of exceptional period
 that this case is possible if a correction message is aborted for
 another correction message).
\end{enumerate}

Note that, taking into account the received information and control
bits, a node is able to update the data of neighbors, i.e., to
determine new values of the lists
$(P(t+1,v): (v,u)\in E(G))$ and $ (R(t+1,v): (v,u)\in E(G))$. In
particular, after receiving the last bit of a correction message from
a neighbor~$v$, a node $u$ sets the participant
 $P(t+1,v)$ to be the
prefix of
$P(t,v)$ of length  $\ell(t,v)$, where the binary representation of an
integer  $\ell(t,v)$ is just the correction message~$L(t,v)$ received.

The last piece of the updated data to be determined is the candidate
of a~node. The new value $C(t+1,u)$ is the longest PL-minimal string
in the set $\{C(t,u), P(t+1,v): (v,u)\in E(G) \}$.

Initially, at moment $t=0$, all participants $P(0,u)$ and all prefixes
of correction messages  $R(0,u)$ are empty and the candidates are the
keys:  $C(0,u) = K(u)$ for all~$u$. 

Note that the candidates, the participants and the correction messages
are functions of the communication history in the node. Thus the
described protocol satisfies the definition~\eqref{evolution}.

Let's check the important property of the protocol, which has been
indicated in the above informal discussion. 

\begin{prop}\label{participant=pref(candidate)}
  A~participant $P(t+1,u)$ is a prefix of some candidate $C(t',u)$, $t'<t$. 
\end{prop}
\begin{proof}
  Induction by the moments of time. Initially, at $t=0$, the
  proposition holds due to the rules of the protocol (the first step
  is regular). 

  Inductive step is a consideration of possible cases for 
  determining  a new value of a participant. 

  A~regular period, the information bit is non-empty. In this case
  $$P(t+1,u)=P(t,u)b(t+1,u)$$ is a prefix of the candidate $C(t,u)$ by
  the rule of determining the information bit.

  A~regular period, the information bit is empty. In this case
  $$P(t+1,u)= C(t,u).$$

  The end of an exceptional period: by definition, $P(t+1,u)$ is a
  prefix of   $C(t,u)$.

  On other steps of an exceptional period, the participant
  does not change. Thus, by the induction hypothesis, it is a prefix
  of a previous candidate.
\end{proof}

Informally, the candidates and the participants provide a sequence of
approximations to the minimal key. Therefore it is natural to expect
that the sequence is monotone in each node. The strings are compared
in a~partial order. So the only weak monotonicity is possible.  The
next proposition formalizes this intuition.

\begin{prop}\label{monotone0}
  A sequence of the candidates $(C(t,u))$ in a node $u$ is non-increasing;
  moreover,  for $t'>t$ either
  $C(t,u)$ is a prefix of $C(t',u)$ or $C(t',u)\ls C(t,u)$.
\end{prop}
\begin{proof}
  By the rule of determining the next candidate as the longest
  possible PL-minimal string over the set that contains the previous
  candidate, it follows from  $C(t+1,u)\not\ls
  C(t,u)$ that  $C(t,u)$ is a prefix $C(t+1,u)$ due to
  Proposition~\ref{min-PL}. 

  For arbitrary $t'>t$ the statement follows by induction and
  transitivity of an order relation.
\end{proof}

\subsection{Correctness of the protocol}

It follows from the rules of the protocol that a node sends the empty
message if  $P(t,u)=C(t,u)$. This condition remains true if the node
has not received non-empty messages from its neighbor (the data remain
unchanged). Therefore, if $P(t,u)=C(t,u)$, then the node $u$ is in a
sleep state.

Below we prove the lemma.

\begin{lemma}\label{weak-correctness}
  At some moment of time the candidates and the participants in all nodes are
  equal to the PL-minimal key in a network.  
\end{lemma}

Proposition~\ref{lex-consistent} (the PL order on the keys is
consistent with the short-lex order on the identifiers), injectivity
of the mapping from the identifiers to the keys, and
Lemma~\ref{weak-correctness} imply that the protocol described in the
previous subsection is a message terminating protocol broadcasting the
minimal identifier.

Before the proof of Lemma~\ref{weak-correctness} we prove several
facts that will also be useful in evaluating of the running time of the
protocol. 

\begin{prop}\label{participant=pref(key)}
  At any moment of time $t$ the participant  $P(t,u)$ and the candidate
  $C(t,u)$ of any node $u$ are prefixes of node keys.
\end{prop}
\begin{proof}
  Induction by the moments of time. It is sufficient to prove the
  statement for candidates only because participants are prefixes of
  candidates due to Proposition~\ref{participant=pref(candidate)}.

  The base case holds since initially all candidates are keys.

  Show that if the statement holds at moment~$t$, then it also holds
  at moment~$t+1$. 

  Participants of all nodes at moment  $t+1$ are prefixes of
  candidates at earlier times due to
  Proposition~\ref{participant=pref(candidate)}. By the induction
  hypothesis they are prefixes of keys. The candidate $C(t+1,u)$
  coincides 
  either  with $C(t,u)$ or with a participant $P(t+1,v)$.
\end{proof}

Let $K_{\min}$ be the minimal key in a network. We will use the
following fact about prefixes of the keys.

\begin{prop}\label{key-dichotomy}
  Let 
  $p$ be a prefix of a key  $K(w)$. Then either $K_{\min} \ls
  p$ or  $p$ is a prefix of $K_{\min}$.
\end{prop}
\begin{proof}
  Look at the position that determines the result of comparison  $K_{\min}$
  and $K(w)$: 
  $$
  K_{\min} =\lcp(K(w),K_{\min}) 0 k',\quad K(w) =\lcp(K(w),K_{\min}) 1 k''. 
  $$ 
  Let $p$ be  a prefix of $K(w)$ such that  $p$ is not a prefix of
  $K_{\min}$.  Then  $p$ is longer than 
  $\lcp(K(w),K_{\min})$. Therefore $\lcp(p,K_{\min})=
  \lcp(K(w),K_{\min})$ and $p =\lcp(K(w),K_{\min}) 1 p' $. Thus
  $K_{\min}\ls p$. 
\end{proof}

\begin{cor}\label{Kmin-stability}
  If $C(t^*,u)=K_{\min}$, then $C(t,u)=K_{\min}$ for all
  $t> t^*$, $u\in V$.
\end{cor}
\begin{proof} 
  Show that  $C(t,u)=K_{\min}$ implies
  $C(t+1,u)=K_{\min}$.  By definition, $C(t+1,u)$ is the
  longest string among all PL-minimal strings in a set containing the
  candidate  $C(t,u)=K_{\min}$ and participants. The participants are
  prefixes of keys due to Proposition~\ref{participant=pref(key)}.
  Now apply Proposition~\ref{key-dichotomy}.
\end{proof}

We conclude from this corollary the following fact: if the candidate
of a node~$u$ equals the minimal key at moment $t^*$, then the node
$u$ will finish communication at some later moment. Indeed, if the
participant $P(t^*,u)$ is not a prefix of the minimal key, then a
correction message will be transmitted (since the candidate remains
unchanged, this correction message can not be aborted). After that the
node resumes transmission of remaining bits of the candidate (which
coincides with the minimal key). The node finishes communication after
sending the last bit of the candidate (the minimal key).

\begin{proof}[Proof of Lemma~\ref{weak-correctness}.]
  Due to Proposition~\ref{Kmin-stability}  it is sufficient to prove
  that in any node the candidate equals the minimal key at some moment
  of time.

  The proof is by induction on the distance $d$  from the node with
  the minimal key. 

  The base case $d=0$ is clear. In this case the candidate is the
  minimal key from the very beginning.

  For induction step, assume that at time $t$ in all nodes at distance
  $<d$ from the minimal one the candidates are equal to the minimal
  key. Take a node $v$ at distance $d$ from the minimal node.  The
  node $v$ has a neighbor $w$ at distance $d-1$ from the minimal node.
  As shown above, the node $w$ finishes transmission of the minimal
  key at some moment $t'\geq t$. At this moment, the participant
  $P(t',w)$ equals the minimal key. It means that   $C(t',v)$ is
  also the minimal key: from Propositions~\ref{greedy-compare}
  and~\ref{participant=pref(key)} we conclude that a prefix of a key
  is not lesser the minimal key in the PL order; thus the minimal
  elements of the set  $\{C(t'-1,u),
  P(t',v): (v,u)\in E(G) \}$ are prefixes of the minimal key and the
  minimal key is the longest among them. Applying the rule for
  determining the next candidate, we get $C(t',v)= K_{\min}$.
\end{proof}

\begin{remark}
  By similar arguments one can easily get a simple upper bound on the
  running time of the protocol, namely, $O(DL)$, where $D$ is the
  diameter of the network and $L$ is the length of the minimal key.

  In the next subsection we improve this bound.
\end{remark}

\subsection{Improved upper bound of the running time}

Assume that the candidate of a node equals the minimal key at some
moment of time $t$. After that moment the node is transmitting the bits of
the minimal key preceding by at most one  correction
message of the length  $O(\log L)$. For inductive bound similar to the
proof of Lemma~\ref{weak-correctness}, we need a lower bound of the
length of the prefix of the minimal key that has been sent by the node
before time~$t$.

This bound is based on monotonicity of the minimal key propagation
through a network. If a node has sent a prefix of the minimal key,
then the prefix will not be corrected in future.

We state this property formally and prove it. But it is convenient to
divide it in two statements.

\begin{prop}\label{monotone1}
  In each node the length of the longest common prefix of the
  candidate in the node and the minimal key does not decrease.
\end{prop}
\begin{proof}
  Let $q_0 = C(t,u)$, $q_1=C(t+1,u)$  be two subsequent candidates in
  a~node. We prove that 
  $\lcp(K_{\min},q_0)$ is a prefix of $q_1$. 

  If $q_0$ is a prefix of $q_1$, then the statement is trivial. So, we
  assume that the opposite holds. Then  $q_1\ls q_0$ due to
  Proposition~\ref{monotone0}.  

  Let  $p $ be $\lcp(q_0,q_1)$. From the definition of the PL order we
  get 
  \[
  q_0 = p1q',\quad q_1=p0q''.
  \]
  By Proposition~\ref{participant=pref(key)}, the string $q_1$ is a
  prefix of a key $K(w)$. Thus,
  $$K_{\min}\lseq K(w)=q_1q'''=p0q''q'''\ls p1.$$
  Therefore $p1$ is not a prefix of $K_{\min}$. But it implies that
  $$\lcp(K_{\min},q_0) = \lcp(K_{\min}, p1) = \lcp(K_{\min},p)$$ 
  is a prefix of $p$, which is a prefix of 
  $q_1$.
\end{proof}

A similar statement holds for participants. But the proof is more
tricky because participants can be arbitrary truncated.

\begin{prop}\label{monotone2}
    In each node the length of the longest common prefix of
    a~participant in the node and the minimal key does not decrease.
\end{prop}
\begin{proof}
  Let  $q_0 = P(t,u)$, $q_1=P(t+1,u)$  be two subsequent participants in
  a~node. By the rules of the protocol, in a regular period  $q_0$ is a
  prefix of $q_1$ (possibly, $q_0=q_1$). At all moments of an
  exceptional period, with except for the last, 
  $q_0=q_1$  too. This last moment is the only nontrivial case in the
  proof. 

  At the end of an exceptional period the rules of the protocol give us
  \[
  q_1 = \lcp(C(t,u), P(t,u)).
  \]

  Now we prove that   $\lcp(K_{\min},P(t,u))$ is a prefix of
  $q_1$. By Proposition~\ref{participant=pref(candidate)},
  $P(t,u)$ is a prefix of some earlier candidate $C(t',u)$,
  $t'<t$. A~correction message is non-empty iff  $P(t,u)$ is not a
  prefix of  $C(t,u)$. It selects in the dichotomy of
  Proposition~\ref{monotone0} the case  $C(t,u)\ls
  C(t',u)$. Thus, 
  \[
  C(t,u) = p0p', \quad C(t',u) = p1p'',
  \]
  and $p1$ is a prefix of $P(t,u)$ (since $P(t,u)$ is not a prefix of
  $C(t,u)$). In other words, $q_1 = p$. Since the candidate $C(t,u)$
  is a prefix of a key  $K(w)$
  (Proposition~\ref{participant=pref(key)}) and
  \[
     p1\gs C(t,u),\qquad K(w) \gs K_{\min},
  \]
  $p1$ is not a prefix of  $K_{\min}$. Thus,  $\lcp(K_{\min},P(t,u))$
  is a prefix of~$p=q_1$. 
\end{proof}

Now we introduce the main tool in the time analysis of the protocol:
\emph{the delays} of the minimal key transmission. 
By definition, the
delay $\delta(t,v)$ in the node $v$ at time $t$ is $t-|\lcp(K_{\min},
P(t,v))|$, if the node has not finished the minimal key transmission at
moment~$t$, i.e., if $P(t,v)\ne K_{\min}$.  Otherwise, the delay is
undefined. 

By Lemma~\ref{weak-correctness}, the participant of a node equals the
minimal key at some moment of time. Therefore the delays in the node
are bounded. 
Let $\delta(v)$ be the maximal delay in a~node $v$. 

Note that a sequence of the delays  $\delta(t,v)$ in a node is
non-decreasing since the only one information bit is transmitted during
a step of protocol.

Also note that 
if the node
$v$ has not finished  transmission of the minimal key at time $t$, then its
participant $P(t,v)$ has a common prefix of the
length at least $t-\delta(v)$  with the minimal key.

\begin{lemma}\label{delay-lemma}
  If nodes $v$ and $w$ are neighbors in a network, then
  \begin{equation}
  \delta(v)\leq\delta(w)+|\bin(L)|,\label{delay-ineq}
\end{equation}
  where $L$ is the length of the minimal key.
\end{lemma}
\begin{proof}
  Let  $t_0$ be the last moment when a node  $v$ sends a wrong bit, i.e.
  \[
  t_0 = \max \Big(t : \big|\lcp(P(t,v),K_{\min})\big|<|P(t,v)|\ \text{and $P(t,v)$ is a
  prefix of $C(t,v)$}\Big).
  \]
  After that, during an exceptional period   $t_0<t\leq t_1$, correction
  messages are transmitted  (therefore $P(t,v) =P(t_0,v)$). 

  At the last moment of correction transmission $t_1$ the length of
  the participant decreases to 
  $\big|\lcp(P(t_0,v),K_{\min})\big|$, since by
  Proposition~\ref{monotone2} the length of the longest common prefix
  of a participant and the minimal key does not decrease but a new
  participant should be a prefix  of the minimal key.
  
  From a moment  $t_2\geq t_1$ the delays do not increase: the next
  bit of the minimal key is transmitted at the next step. As the delay
  function is monotone, we get
  \[\delta(v) \leq t_2 - \big|\lcp(P(t_0,v),K_{\min})\big|.\]

  We upperbound  $t_2 $ via the delay  $\delta(w)$ of a neighbor $w$
  of the node~$v$. As it shown above, either the participant $P(t,w)$
  has a common prefix with the minimal key that is not shorter than
  $\ell(t) = t-\delta(w)$, or this participant is the minimal key. The
  participant is in the set of strings, which is used in determining
  the next candidate of the node~$v$.   Therefore the candidate
  $C(t_0,v)$ has a common prefix with the minimal key that is not
  shorter than  $\min(\ell(t_0),L)$. On the other hand, this common
  prefix coincides with  $\lcp(P(t_0,v),K_{\min})\ne K_{\min}$ by
  construction: the moment $t_0$ is regular, thus $P(t_0,v)$ is a
  prefix of  $C(t_0,v)$. From this we get
  $$
  \big|\lcp(P(t_0,v),K_{\min})\big| \geq \ell(t_0),
  $$
  i.e., $t_0\leq\big|\lcp(P(t_0,v),K_{\min})\big|+\delta(w)$. 

  Note that for 
  $t\geq \big|\lcp(P(t_0,v),K_{\min})\big|+\delta(w)$ the length of
  the candidate increases at step  $t$ while the candidate differs
  from the minimal key. If  $t\geq t_1$ holds, then the length of the
  participant also increases while the participant   differs
  from the minimal key. Thus
  \[
  t_2\leq \max\Big( \big|\lcp(P(t_0,v),K_{\min})\big|+\delta(w), t_1\Big).
  \]

  The difference  $t_1-t_0$ equals the length of the last correction
  message 
  $$\bin\Big(\big|\lcp(P(t_0,v),K_{\min})\big|\Big),$$
  which is at most $|\bin(L)|$.

  Finally, we get
  \begin{multline*}
  \delta(v) \leq  t_2 - \big|\lcp(P(t_0,v),K_{\min})\big|\leq\\
  \leq
  \max\Big( \big|\lcp(P(t_0,v),K_{\min})\big|+\delta(w), t_1\Big) -
  \big|\lcp(P(t_0,v),K_{\min})\big| \leq \\ \leq
  \max\Big(\delta(w), t_1-t_0
  +t_0-\big|\lcp(P(t_0,v),K_{\min})\big|\Big)
  \leq   \delta(w) +|\bin(L)|.
  \end{multline*}
  It completes the proof.
\end{proof}

\begin{lemma}
  The running time of the protocol is  $ O(L+D\log L)$.
\end{lemma}
\begin{proof}
  The delay in the node with the minimal key is $0$. 
  Inequality~\eqref{delay-ineq} implies that for a node at distance
  $h$ from the minimal node the delay is $O(h\log L)$. So the maximal
  delay over all nodes is $O(D\log L)$. Thus each node finishes 
  transmission of the minimal key in time $ O(L+D\log L)$ and the
  running time of the protocol has the same bound.
\end{proof}

\subsection{Spanning tree construction}\label{span-tree}

To  complete the proof of Theorem~\ref{th:m-t-prot},
we need to show that the above protocol constructs a spanning tree.

After finishing the protocol, all nodes are in sleep states and their
participants and candidates are equal to the minimal key. One can
define a spanning tree at this moment by the following rule: a node
looks at the moment, when it has received the last bit of the minimal key,
and chooses as a parent a node that has sent this bit (a choice between
several nodes satisfying this condition is
arbitrary).  The node having the minimal key is the root. We obtain in
this way a directed graph (a digraph). Fan-outs of all nodes in the digraph
$\leq1$ and there are no directed cycles in it
(to send the last bit of the minimal key, the node should receive it;
thus the moments of transition to a sleep state form a strictly
decreasing sequence). 

But for processor terminating protocol we need a more general
construction. Not only the minimal key is propagating through  a
network. It is possible that at some moment the protocol is stopped in
a part of the network but the candidates in this part differ from the
minimal key. Recall that no key is a~prefix of another
key. By Proposition~\ref{participant=pref(key)} candidates are
prefixes of keys. Thus a~node can distinguish the moments of time when
its candidate is a key as well as neighbors who sent the last bit of
the key.

For a node $v$ we define a sequence of keys
\begin{equation}\label{key-seq}
K(v) = K_0\gs K_1\gs \dots\gs K_{\text{\tiny last}} = K_{\min},
\end{equation}
passing through the node at some moment of time. The
sequence~\eqref{key-seq} is decreasing due to
Proposition~\ref{monotone0} (recall that the PL order is a linear
order on the set of keys). The sequence is non-empty, since initially
the candidate is the key of the node.

For a key $K$ we denote by $\tau_K(v)$ the first moment $t$ such that
$K = C(t,v)$. If $K\ne C(t,v)$ for all $t$, then $\tau_K(v) = \infty$.
By definition,  $\tau_{K(v)}(v)=0$ and other values of 
$\tau_{K}(v)$ are either positive or infinite. 

If  $K\ne K(v)$ and $\tau_K(v)<\infty$, then $C(\tau_K(v)-1,v)\ne
C(\tau_K(v),v)=K$. By the rule of determining the candidate, we get
 $C(\tau_K(v),v) = P(\tau_K(v),w_K(v))$ for some neighbor~$w_K(v)$ of the 
node~$v$. Actually, it might be several neighbors satisfying this
condition. Let choose one of them  arbitrary. We say that the
node 
$v$ receives  (the last bit of) the key $K$ from the node~$w_K(v)$. 

\begin{prop}\label{birth-seq}
  $\tau_K(w_K(v))<\tau_K(v)$.
\end{prop}
\begin{proof}
  The sequence of candidates is non-increasing in the PL
  order. A~participant of a node is a prefix of an earlier candidate of the
  node  (Proposition~\ref{participant=pref(candidate)}). Therefore the
  key $K$ is contained in the sequence of the keys passing through the
  node  $w_K(v)$. Moreover, it follows from 
  $C(\tau_K(v),v) = P(\tau_K(v),w_K(v))$ that $\tau_K(w_K(v))<\tau_K(v)$.
\end{proof}

Note that  $w_K(v)$ is a function of the communication history 
$H(v,t)$
in the node  $v$ for  $t\geq \tau_K(v)$.

Thus, to define a function  $\pi(H)$ from the definition of a protocol
constructing a spanning tree, we set  $\pi(H(v,t)):=w_K(v)$, where $K$
is the minimal key among the keys that has passed through the node  $v$ up
to the moment~$t$. Formally, 
$K=K_i$ and $K_i$ is in the key sequence for the node $v$ such that
$\tau_K(v)\leq 
t$ and if  $K\ne K_{\min}$, then $t< \tau_{K_{i+1}}(v)$.  If 
$K=K(v)$, then $\pi(H(v,t))=\root$.

At the moment $T$, when the protocol is finished, the candidate of any
node is the minimal key $K_{\min}$. For the node $r$ having the
minimal key, the function $\pi(H(r,T))$ defined above has the value
$\root$. In other nodes $\pi(H(v,T))$ points to
$w_{K_{\min}}(v)$. Thus $\pi(H(v,T))$ defines a digraph and fan-outs
of all nodes are $\leq1$. Along a directed path in the digraph, the
values $\tau_{K_{\min}}(v)$ are strictly decreasing
(Proposition~\ref{birth-seq}). So, there are no cycles in the digraph
and it is a rooted tree.

This argument completes the proof of Theorem~\ref{th:m-t-prot}.

\section{Processor terminating protocols}\label{finalizing}

In this section we prove Theorem~\ref{th:p-t-prot}. To define processor
terminating protocols that broadcast the minimal identifier and
construct a spanning tree we use the function $\pi(H)$ described in
the Subsection~\ref{span-tree} and a structure relying on this
function. For the proofs we also need auxiliary trees in a
network. These trees are defined in the next subsection. After that we
provide a protocol description and the proof of correctness.

\subsection{Dynamic and static trees}

An analysis of a protocol proposed below is based on a relation
$u=w_K(v)$ ``a~node $v$ receives a key $K$ from a node~$u$'' and
changes of the relation in time. For this purpose we define a digraph
$\Gamma(t)$ including in it all edges in the form $(v, \pi(H(v,t)))$,
$\pi(H(v,t))\ne\root$.

Let  $(v,u)$ be an edge of the digraph $\Gamma(t)$. It means by
definition that  $u = w_K(v)$ for some key  $K$ and the key $K$ is the
minimal (in the PL order) key that has passed through the node~$v$ up to a
moment~$t$. 
We say that  $K$ is the 
\emph{edge key} and  $\tau_K(v)$ is the
\emph{birth date} of the edge.

Initially,   $\Gamma(0)$ consists of isolated roots. At the end, as it
was shown above, it is a spanning tree. For the proof of correctness
we will need a stronger property.

\begin{lemma}\label{forest-lemma}
  The digraph  $\Gamma(t)$ is a spanning forest for any  $t$.
\end{lemma}
\begin{proof}
  The fan-outs of the nodes are  $\leq1$ in this digraph by
  construction. Therefore the only cycles in the underlying undirected
  graph are the directed cycles in the digraph  $\Gamma(t)$. 
  To prove that there are no directed cycles we choose an appropriate
  monotonicity property.

  A~transition from an edge $(v,u)$ to an edge 
  $(u,w)$ in the digraph  $\Gamma(t)$  decreases either the edge key
  (the sequence of keys is monotone in each node) or the birth date
  (Proposition~\ref{birth-seq}).

  Thus the lengths of directed paths in $\Gamma(t)$ are bounded and
  the digraph does not contain directed cycles.
\end{proof}

Connected components of the digraph 
$\Gamma(t)$ are rooted trees  (the roots are nodes with fan-out~0). 
We call them  \emph{dynamic trees} in the network. By  $\Gamma(t,K)$
we denote the connected component that has a root $r$ with the key 
$K=K(r)$.

In the proof below we will use also  \emph{static} trees corresponding
to keys. Let $K=K(r)$. Then the vertices of the rooted tree 
$T_K$ are the nodes such that $\tau_K(v)<\infty$. The node~$r$ is the
root, a~parent of a node $v$ is the node
 $w_K(v)$. Since $\tau_K(v)$ is decreasing on a move from a child
to a parent, the graph  $T_K$ is a tree indeed.
We call these tree static because they are defined by the whole
communication history. For some keys the tree  $T_K$ consists of the
root with the key~$K$ only. For the minimal key $K_{\min}$
the static tree 
$T_{K_{\min}}$ coincides with the spanning tree constructed by the message
terminating protocol from the previous section.

On Figure~\ref{Gt} two components of a dynamic tree
$\Gamma(t)$ are shown. Figures~\ref{TK1}--\ref{TK2} show the
corresponding parts of static trees  $T_{K_1}$ and $T_{K_2}$. The
roots of the trees are white circles. Links of the network that are not
included in trees are pictured by dashed lines.
 
\begin{figure}
\noindent
  \begin{minipage}[b]{0.34\textwidth}
    \centerline{\mpfile{trees}{1}}
    \caption{$\Gamma(t)$, $K_2\ls K_1$}\label{Gt}
  \end{minipage}\hfill
  \begin{minipage}[b]{0.31\textwidth}
    \centerline{\mpfile{trees}{2}}
    \caption{$T_{K_1}$}\label{TK1}
  \end{minipage}\hfill
  \begin{minipage}[b]{0.31\textwidth}
     \centerline{\mpfile{trees}{3}}
    \caption{$T_{K_2}$}\label{TK2}
  \end{minipage}\hfill
  \parfillskip=0pt\par
\end{figure}

It follows from the definition that a node $v$ is included in a tree
$T_{K}$ iff the candidate of the node equals the key~$K$ at some
moment. 

\begin{remark}\label{conserve}
Moreover, 
an edge $(v,u)$ of a static tree $T_K$ is in the digraph
 $\Gamma(t)$ at all moments $t$ such that the key $K$ is the
PL-minimal key among the keys passed through the node~$v$. This remark
will be used in the proof below.
\end{remark}

\subsection{Description of the processor terminating protocol}

The required protocol is combined from three protocols $A$, $B$,
$C$. Its operation is divided in stages. During each stage the
combined protocol performs one step of each protocol
 $A$, $B$,
$C$ in the indicated order. Thus, each protocol is operating 3 times slower. It
does not affect the asymptotic bound in Theorem~\ref{th:p-t-prot}. 
In the sequel, we refer to stages of the combined protocol. After
performing  $t$ stages of the combined protocol each protocol $A$,
$B$, $C$ has performed  $t$ steps.

The protocol $A$ is the message terminating protocol described in the previous
section. It broadcasts the minimal identifier (formally, it broadcasts
the minimal key, but an identifier can be recovered by a key).

The second protocol $B$ maintains the spanning forest $\Gamma(t)$
defined in Subsection~\ref{span-tree}. 
In the protocol $B$ nodes send messages of two types: ``I'm a child''
and ``I'm not a child''. 

The first message is sent by a node $v$ at that moment when its
candidate $C(t,v)$ becomes a key $K\ne K(v)$. The node $v$ sends this
message to the neighbor $w_K(v)$, which sent the last bit of the
key~$K$. If the node $v$ has a parent $w\ne w_K(v)$ at this moment
(i.e., in the tree $\Gamma(t-1)$), then the node $v$ also sends the
message of the second type to the node~$w$.  Otherwise, the node sends
no message (formally, the empty message).

Since a step of the protocol $A$ is followed by a step of the protocol $B$,
before a step of the protocol $C$ 
nodes have correct information about their childs in the forest
$\Gamma(t)$.

In the finalizing protocol $C$ each node $v$ checks the \emph{local
termination condition}: 
communication with the neighbors in the protocol $A$ is
finished, the empty message is received in the protocol $B$, and the
candidate  and all participants of the node $v$ are equal to some key~$K$.

In the protocol $C$ nodes send messages through the edges of the
spanning forest $\Gamma(t)$. There are two types of messages: the confirmation (of
local finish) message and the termination message. A~node also can
send no message.

Divide nodes in a dynamic tree 
$\Gamma(t,K)$ 
into three
groups: the root has fan-out~$0$, the intermediate nodes have childs
(fan-in $>0$) 
and the leaves have no childs (fan-in $0$). 

Communication in
the protocol $C$ starts in leaves. A~leaf sends the confirmation
message to its parent if the local termination condition becames true
for the leaf.

An intermediate node collects confirmation messages from its
childs. It sends the confirmation message to its parent if it has
received confirmations from all its childs and the node itself
satisfies the local termination condition.

It is possible that the local termination condition changes the value
from true to false: a node exchanged a key with all neighbors but the
key is not the minimal one. If  the candidate of a node is decreased in
the PL order after a step of the protocol $A$
or the node receives the message ``I'm not a child'' after a step of
the protocol $B$, then the node removes all confirmations received. In
further operation, it
ignores all incoming confirmations until its candidate becomes a key. 

The root also collects confirmations from its childs and checks the
local termination condition.  The root sends the
termination message to all its childs if all childs of the root have
sent confirmations and the root satisfies the local termination
condition. 

The termination messages broadcast from parents to childs
through the tree. After
sending the termination message a node goes to the final state.

\subsection{Correctness of the combined protocol}

\begin{lemma}
  The combined protocol is processor terminating: \textup{(1)} if a
  node is in the 
  final state at time $t$,
  then the protocol $A$ has finished at time $t$; \textup{(2)} all nodes are in the final
  state at some moment of time.
\end{lemma}
\begin{proof}
  To prove the first statement of the lemma, we show that if a root of a
  dynamic tree  $\Gamma(t,K)$ sends the termination message, then the
  tree is spanning.  

  We say that a confirmation is passed through an edge  $(v,u)$ with
  the edge key $K'$ if at some moment 
  $ \tau_{K'}(v)\leq t'<t$ the node  $v$ sends the confirmation
  message to the node~$u$.

  Nodes remove confirmations if their candidates decrease or the
  message ``I'm not a child'' is received. Thus it is sufficient to
  take into account only those confirmations that were passed through
  the edges of the tree $\Gamma(t,K)$ with the current edge keys 
  (but an edge can be in an
  another tree  at the moment of passing confirmation). 
  
  We conclude from this observation that the root does not send the
  termination message at moment $t$ if there exists an edge  $(v,u)$
  of the tree
  $\Gamma(t,K)$ such that no confirmation has  passed through the edge.
  Indeed, the keys in the roots of trees containing the edge $(v,u)$
  form a non-increasing sequence. Therefore, while the edge  $(v,u)$
  is in the tree
  $\Gamma(t,K)$, this tree also contains all edges of the path from
  $u$ to the root and no confirmation was passed through edges of this
  path. It means that the root does not send the termination message. 

  If the tree $\Gamma(t,K)$ with the root $r$ is not a spanning tree,
  then the following cases are possible.
  \begin{enumerate}
  \item\label{cand-less} $C(t,r)\ls K = K(r)$.
  \item\label{key-geq} Otherwise, $C(t,r) = K$ due to monotonicity of
    the candidates in a node. But it is possible that the tree
    $\Gamma(t,K)$ includes edges whose keys differ from~$K$. 
  \item\label{good-neighbors} $C(t,r) = K$ and the edge key is $K$ for
    each edge in the tree $\Gamma(t,K)$. We assume that the tree is
    not spanning. Thus, there exist ``external links'': between nodes
    in the tree and outside it (the network itself assumed to be
    connected). We call \emph{tree neighbors} the  endpoints of the
    second type.

    Suppose that the following condition holds: there is no neighbor
    of the tree  $\Gamma(t,K)$ such that its candidate is $K$ at some
    earlier moment.  
  \item\label{bad-neighbors} The last case covers the remaining
    option. Namely, there exists a neighbor of the tree  $\Gamma(t,K)$
    such that its candidate is $K$ at some
    earlier moment.  
  \end{enumerate}

\noindent\textbf{Case~\ref{cand-less}.}  The local termination
  condition violates in the root $r$ in this case: the candidate
  differs from a key. Thus the root does not send the termination
  message. 

For all other cases we indicate an edge in the tree 
$\Gamma(t,K)$ such that no confirmation has been passed through the edge.

\noindent\textbf{Case~\ref{key-geq}.}  
  By definition of a digraph  $\Gamma(t)$, the root $r$ has not been
  in any tree  $\Gamma(t',K')$, $t'<t$, $K'\ne K$. It implies that the
  only possible value of the edge key for an edge $(u,r)$  is $K$.

  Thus, there are two edges in the tree $(u,v)$, $(v,w)$ such that the
  first edge key is $K'\ne K$, the second edge key is~$K$, and $K\ls
  K'$. Indeed, both keys passed through the node $v$ up to the moment
  $t$ and inequality $K\gs K'$ contradicts the definition of a digraph
  $\Gamma(t)$.

  Show that no confirmation has passed through the edge $(v,w)$.
  Suppose the opposite: the node $v$ has sent the confirmation message
  to the node~$w$ at a moment $\tau_{K}(v)\leq t'<t$. The local
  termination condition implies  $C(t',u)= C(t',v) = C(t', w) =
  K''$. But  $K''=K$, otherwise the node $v$ has removed this
  confirmation up to the moment~$t$.

  It implies that the key sequence~\eqref{key-seq} in the node $u$
  contains the key $K\ls K'$ and the edge key for the edge $(u,v)$ is
  less than~$K'$. It contradicts the assumption. 

  This contradiction completes the argument for this case.

\noindent\textbf{Case~\ref{good-neighbors}.}  
  Suppose that a confirmation passes through an edge  $(u,v)$. Then
  the node $u$ satisfies the local termination condition. It implies
  that candidates of all neighbors are~$K$. 

  From this observation we conclude that nodes in the
  tree~$\Gamma(t,K)$ having ``external links'' have not sent the
  confirmations. (Recall that in the case into consideration 
  candidates of any neighbor
  of the tree  $\Gamma(t,K)$ differ from $K$ at any moment up to~$t$.)

\noindent\textbf{Case~\ref{bad-neighbors}.}  Recall that in this case 
    there are nodes outside  the tree  $\Gamma(t,K)$
    such that their candidates were $K$ at some
    earlier moments.  All these nodes are in the static tree~$T_K$.  
    Choose among them the neighbor $w$ of the tree  $\Gamma(t,K)$
    that has leaved the dynamic tree with the key $K$ at the earliest
    moment. Let 
    \begin{equation}\label{bad-path}
      w=v_0, v_1, \dots, v_k = r
    \end{equation}
    be the path in the static tree $T_K$ from $w$ to the root
    $r$. Show that no node on the path has sent the confirmation
    through an edge of this path (with the edge key $K$). 
    
    While the node $w$ is in the dynamic tree with the key $K$, the
    conditions of the case~\ref{good-neighbors} hold: there is no neighbor of
    the tree such that a~candidate of the neighbor was $K$.  It
    guarantees that the node $w=v_0$ has not sent confirmation.

    Suppose that confirmations were sent by nodes on the path from $w$
    to $r$ at earlier moments than $t$.  Among the nodes, which sent
    these confirmations, take $v_i$
    --- the nearest to the node $w$. As it was shown above, the node
    $v_i$ differs from $w$. Therefore there is the preceding node
    $v_{i-1}$ on the path.

    If a node $v_{i-1}$ is a child of a node $v_i$ in $\Gamma(t')$,
    then the node $v_i$ has not received the confirmation from
    $v_{i-1}$. 
    
    Otherwise, the candidate $C(t',v_{i-1})$ is lesser than 
    $K$ in the PL order. It implies that either the node $v_i$  does not
    finish communication or it received the message ``I'm not a child''
    from  $v_{i-1}$.

    Both variants  contradict the sending of confirmation by the
    node~$v_i$.

    So, no confirmation has been passed through the edges of the
    path~\eqref{bad-path}. Take the minimal key $K'\lseq K$ that has
    been passed through the node $v_{k-1}$. If $K'=K$, then the edge
    $(v_{k-1}, v_k)= (v_{k-1}, r)$ is in the tree $\Gamma(t,K)$ (see
    Remark~\ref{conserve}) and  no confirmation has been passed
    through the edge. Otherwise, the local termination condition does
    not hold in the root $r$ because the candidate of $v_{k-1}$ is
    lesser than $K$.

    This completes our case analysis. So, we have proved that if the tree
    $\Gamma(t,K)$ is not spanning, then its root does not send the
    termination message.  

    To prove the second statement of the lemma, note that 
    the protocol $A$ terminates at some time $t'$.
    The graph $\Gamma(t')$ is a spanning tree. All candidates $C(t',v)$
    and all participants $P(t',v)$ coincide with the minimal key. Thus all
    nodes finish communication in the protocol~$A$. 
    It means that in the protocol $C$ the
    confirmations are transmitted upward. When the root collects the
    confirmations from all its childs, it sends the termination
    message. The termination message will broadcast downward.
\end{proof}

\begin{proof}[Proof of Theorem \ref{th:p-t-prot}.]
  The running time of the protocol $A$ is  $O(L+D\log L)$ due to
  Theorem~\ref{th:m-t-prot}. After that all nodes have finished the
  communication in the protocols $A$, $B$ and confirmations will broadcast
  through the spanning tree without delays as well as termination
  messages. So, an additional time to finish the combined protocol is 
  $O(d)$, where $d$ is the depth of the spanning tree constructed.

  To bound  the depth of the tree, note that
  all nodes receive the minimal key in time $ O(L+D\log L)$. 
  The depth is not greater than this value. Indeed,
  the moments of receiving the minimal key are different along a
  path from the root to a node. After receiving the minimal key the node
  does not change a parent.

  Thus, the running time of the combined protocol is
  \[
  O(L+D\log L)+ O(d) = O(L+D\log L).
  \]
\end{proof}

\section*{Acknowledgments}

The authors are thankful to anonymous referees for their numerous
valuable comments and suggestions.

\end{document}